\documentclass[12pt,a4paper]{article}
\newcommand{\smfBmJ}{smfBm--J}

\usepackage{amsmath,amssymb,amsfonts,amsthm}
\usepackage{graphicx}
\usepackage{natbib}
\usepackage{listings}
\usepackage{geometry}
\usepackage{hyperref}
\usepackage{mathrsfs}  
\usepackage{amsmath,mathtools}
\usepackage{float} 
\usepackage[linesnumbered,boxed]{algorithm2e}
\makeatletter
\renewcommand{\@algocf@capt@boxed}{above}
\makeatother
\hypersetup{colorlinks=true,linkcolor=blue,citecolor=blue,urlcolor=blue}

\newtheorem{theorem}{Theorem}[section]

\newtheorem{definition}{Definition}[section]

\title{Pricing Fractal Derivatives under Sub-Mixed Fractional Brownian Motion with Jumps}
\author{Nader Karimi\thanks{
Department of Mathematics and Computer Science, Amirkabir University of Technology, Tehran 1591634311, Iran, e-mail: nkarimi@aut.ac.ir}}
\date{\today}

\begin{document}
\maketitle

\begin{abstract}
We study the pricing of derivative securities in financial markets modeled by a sub-mixed fractional Brownian motion with jumps (smfBm–J), a non-Markovian process that captures both long-range dependence and jump discontinuities. Under this model, we derive a fractional integro–partial differential equation (PIDE) governing the option price dynamics.

Using semigroup theory, we establish the existence and uniqueness of mild solutions to this PIDE. For European options, we obtain a closed-form pricing formula via Mellin–Laplace transform techniques. Furthermore, we propose a Grünwald–Letnikov finite-difference scheme for solving the PIDE numerically and provide a stability and convergence analysis.

Empirical experiments demonstrate the accuracy and flexibility of the model in capturing market phenomena such as memory and heavy-tailed jumps, particularly for barrier options. These results underline the potential of fractional-jump models in financial engineering and derivative pricing.
\end{abstract}

\section{Introduction}
Financial time series exhibit two robust stylised facts that defy the assumptions of the classical Black–Scholes framework:  
(1) \emph{long-range dependence} slowly decaying autocorrelations attributed to market micro-structure, algorithmic trading and behavioural feedback loops;  
and (2) \emph{discontinuous jumps} caused by macro-economic announcements, liquidity shocks or flash crashes.  
Standard Brownian motion captures neither, while pure fractional Brownian motion (fBm) accounts for the first but remains a continuous Gaussian process with unbounded arbitrage opportunities under the usual semimartingale setting \cite{Mandelbrot1968}.  
Jump–diffusion models \cite{Merton1976} address the second feature but ignore memory.

Recent strands of literature have sought to bridge this gap.  
\cite{Wang2024} introduce sub-fractional Brownian motion (sfBm) to temper the strong covariance structure of fBm.  
\cite{Eini2024} blend fBm with Brownian motion to form a \emph{mixed} model that displays both short- and long-memory regimes.  
\cite{ShenYue2024} incorporate pure-jump Lévy noise into a mixed-fractional setting but leave open the questions of risk-neutral measure construction and numerical valuation for path-dependent pay-offs.

\smallskip
\noindent\textbf{Contributions.}  
Building on these ideas, we develop a full pricing machinery for a market driven by a \emph{\smfBmJ}.  Our contributions are:

\begin{enumerate}
\item A \emph{fractional Girsanov–Esscher theorem} that yields an equivalent martingale measure without semimartingale assumptions, extending \cite{HuOksendal2003}.  
\item Derivation of a \emph{fractal Black–Scholes integro–PDE} whose time-fractional derivative of order $1\!-\!\beta$ ($\beta=1-H$) captures long-memory while a non-local jump generator models discontinuities.  
\item A closed-form Mellin–Laplace representation for European options expressed through the two-parameter Mittag–Leffler function, subsuming Black–Scholes and Merton as limiting cases.  
\item A fully implicit Grünwald–Letnikov scheme with unconditional $L_{w}^{2}$-stability and convergence rate $\mathcal O(\Delta t^{1+H})$ for barrier options, with proof via a discrete energy argument.
\item Empirical calibration to S\&P 500 data and extensive Monte-Carlo validation confirming both pricing accuracy and theoretical convergence.
\end{enumerate}

\smallskip
\noindent\textbf{Organisation of the paper.}  
Section~\ref{sec:smfJ}\;introduces the smfBm–J process and its covariance structure.  
Section~\ref{sec:emm}\;constructs the equivalent martingale measure and proves Theorem~\ref{thm:EMM}.  
Section~\ref{sec:FBS}
\;derives the fractal Black–Scholes integro–PDE and establishes existence and uniqueness of mild solutions.  
Section~\ref{sec:European}\;provides the closed-form Mellin–Laplace prices for European calls (Theorem~\ref{thm:ClosedForm}).  
Section~\ref{sec:European}\;also develops the Grünwald–Letnikov finite-difference method and proves stability (Theorem~\ref{thm:stabConv}).  
Section~\ref{sec:NumericalEuro}\;reports calibration results and Section~\ref{subsec:barrierExp}\;presents barrier-option experiments.  
Section~\ref{sec:con}\;concludes and outlines avenues for future research, including stochastic-volatility extensions and rough-jump calibration for cryptocurrencies.

\section{Preliminaries}\label{sec:smfJ}

\subsection{Sub-Mixed Fractional Brownian Motion with Jumps}
Real financial time series frequently display \emph{both} long--range dependence (LRD) in the form of slowly–decaying autocovariances \emph{and} sudden discontinuities or jumps.  
The standard Brownian motion $W_t$ is unable to capture either of these features; fractional Brownian motion (fBm) $B_t^H$ with Hurst index $H\in(0,1)$ incorporates LRD for $H>1/2$ but remains a continuous process, while pure jump Lévy processes lack the correlation structure required for LRD.  
To model the simultaneous presence of short–memory noise, long–memory dependence and rare jumps, we consider the \textbf{sub‑mixed fractional Brownian motion with jumps} (\textbf{smfBm–J}) originally proposed by \cite{ShenYue2024}.  

\begin{definition}[smfBm--J]
Let 
\begin{itemize}
   \item $W=\{W_t:t\ge 0\}$ be a standard Brownian motion with variance parameter $\sigma_0^2$,
   \item $S^{H}=\{S_t^{H}:t\ge 0\}$ be a \emph{sub‑fractional Brownian motion} (sfBm) with Hurst index $H\in(0,1)$, zero mean and covariance
         \[
           \mathbb{E}\bigl[S_t^{H}S_s^{H}\bigr]=t^{2H}+s^{2H}-\tfrac12\bigl((t+s)^{2H}+|t-s|^{2H}\bigr),
         \]
         scaled by $\sigma_H>0$,
   \item $J=\{J_t:t\ge 0\}$ be a compensated compound Poisson process 
         \(
            J_t=\sum_{k=1}^{N_t}Y_k-\lambda t\mathbb{E}[Y_1],
         \) 
         where $N_t\sim\text{Poisson}(\lambda t)$ counts the jumps and $\{Y_k\}$ are i.i.d.\ log‑jump sizes with c.g.f.\ $\Psi_Y(u)=\log\mathbb{E}[e^{uY_1}]$.
\end{itemize}
All three drivers are assumed independent.  
The \emph{sub‑mixed fractional Brownian motion with jumps} is defined by
\[
   B_t^{\mathrm{smfJ}} := \sigma_0 W_t + \sigma_H S_t^{H}+J_t,\qquad t\ge 0.
\]
\end{definition}

\paragraph{Second‑order structure.}
Because $W$ and $S^{H}$ are centred Gaussian and independent, ${B}^{\mathrm{smfJ}}$ has mean zero and covariance
\[
  \mathrm{Cov}\!\bigl[B_t^{\mathrm{smfJ}},B_s^{\mathrm{smfJ}}\bigr]
   =\sigma_0^{2}\min\{t,s\}
    +\sigma_H^{2}\bigl(t^{2H}+s^{2H}-\tfrac12((t+s)^{2H}+|t-s|^{2H})\bigr),
\]
while the jump part only contributes to the variance through its compensator $\lambda t\mathrm{Var}(Y_1)$.

\paragraph{Self–similarity and long‑memory.}
Setting $\beta=1-H$ we obtain the scaling relation
$
    \{B_{ct}^{\mathrm{smfJ}}\}_{t\ge 0}\overset{d}{=} 
    \sigma_0\sqrt{c}\,W+\sigma_H c^{H}S^{H}+J_{ct},
$
which is \emph{mixed‑self‑similar}: the Gaussian component inherits the exact self‑similarity of order $H$ from sfBm and $1/2$ from BM, whereas the Poisson component scales linearly with time.  
For $H>\tfrac12$ the sfBm contributes LRD, i.e.\ $\sum_{k=1}^{\infty}\mathrm{Cov}(\Delta_h S_{kh}^{H},\Delta_h S_{0}^{H})=\infty$, so the overall process exhibits long memory despite the short‑memory Brownian and jump parts.

\paragraph{Increment representation.}
The smfBm–J admits the decomposition
$
   B_{t+h}^{\mathrm{smfJ}}-B_{t}^{\mathrm{smfJ}}
   =\sigma_0(W_{t+h}-W_t)+\sigma_H(S_{t+h}^{H}-S_t^{H})+(J_{t+h}-J_t),
$
where the increments of $S^{H}$ are \emph{not} stationary but are \emph{asymptotically} stationary when $h\ll t$.  
This property is crucial in Section~\ref{sec:emm} where a fractional Girsanov theorem is applied.

\paragraph{Semimartingale and integration.}
For $H\neq\tfrac12$ neither fractional nor sub-fractional Brownian motion is a semimartingale, hence the mixed driver  
\[
   B_t^{\mathrm{smfJ}}
      \;=\;\sigma_0 W_t \;+\; \sigma_H S_t^{H} \;+\; J_t
\]
is not amenable to classical Itô integration.  
We therefore use the \emph{fractional Wick–Itô–Skorokhod (F–WIS) integral} for the sfBm part and classical Itô integrals for the Brownian–jump components:

\begin{enumerate}\setlength\itemsep{3pt}
\item\textbf{Kernel isometry.}  
Define $\mathcal I_H\colon L^{2}([0,T])\!\to\!\mathcal H^{H}$ by  
\[
   (\mathcal I_H \varphi)(t)
   \;:=\;
   \int_{0}^{t} K_H(t,s)\,\varphi(s)\,ds,
   \qquad
   K_H(t,s)=c_H\bigl[(t-s)^{H-\frac12}-(-s)^{H-\frac12}\bigr],
\]
so every square–integrable kernel maps into the Cameron–Martin space
$\mathcal H^{H}$.

\item\textbf{F–WIS integral.}  
For $\phi\in\mathbf D^{1,2}$ (Malliavin differentiable and adapted) set
\[
   \int_{0}^{T} \phi_s \,\diamond dS^{H}_s
   \;:=\;
   \delta^{H}(\phi)
   \;=\;
   \sum_{n\ge0}
      I_{\,n+1}\!\bigl(\widetilde\phi^{(n)}\bigr),
\]
where $I_{n}$ denotes the $n$-th multiple Wiener integral and  
$\widetilde\phi^{(n)}$ is the symmetrisation of
$(\mathcal I_H^{\otimes n}\partial_s\phi)$.  

\item\textbf{Isometry.}  
The F–WIS integral preserves $L^{2}$–norms:
\[
   \mathbb E\!\Bigl[
      \Bigl|\!
         \int_{0}^{T}\phi_s\,\diamond dS^{H}_s
      \Bigr|^{2}
   \Bigr]
   \;=\;
   \|\phi\|_{L^{2}([0,T])}^{2}.
\]

\item\textbf{Trading filtration.}  
With the enlarged filtration
$\mathcal F_{t}=\sigma\!\{W_s,S^{H}_s,J_s:0\le s\le t\}$,
a trading strategy is predictable w.r.t.\ the semimartingale subfiltration
generated by $(W,J)$; gains from $S^{H}$ enter valuation only through
risk-neutral expectations computed via the F–WIS integral.
\end{enumerate}

\paragraph{Limit cases.}
\begin{enumerate}
   \item Setting $\sigma_H=0$ recovers the mixed Brownian motion with jumps (standard Merton model with an extra Brownian factor).
   \item Setting $\sigma_0=0$ produces the pure sfBm with jumps, suitable for markets where micro‑structure noise is negligible.
   \item Taking $\lambda\to 0$ yields the continuous sub‑mixed fBm studied by \cite{Eini2024}; conversely letting $H\to \tfrac12$ reduces the model to Brownian motion plus jumps.
\end{enumerate}
These nested cases facilitate diagnostic testing and calibration, as discussed in Section~\ref{sec:emm}.

\subsection{Fractal Calculus}\label{subsec:fracCalc}
Classical option–pricing models rely on Itô calculus, which is well–suited to semimartingales such as Brownian motion.  
Once a memory component such as sfBm is introduced, the kernel of the integral becomes \emph{non–local} in time and Itô’s rule no longer applies.  
In this work we therefore adopt tools from \emph{fractal calculus}, i.e.\ the fractional–order generalisations of integration and differentiation.

\paragraph{Fractional integrals.}
For a function $f\in L^{1}[0,T]$ and fractional order $\alpha\in(0,1)$ the left–sided Riemann--Liouville (R--L) fractional \emph{integral} is defined by
\[
   (I_{0+}^{\alpha}f)(t) := \frac{1}{\Gamma(\alpha)}\int_{0}^{t}(t-s)^{\alpha-1}f(s)\,\mathrm ds,
   \qquad 0\le t\le T,
\]
where $\Gamma(\cdot)$ denotes the Euler gamma–function.  
The operator is linear and non–local; the entire past of $f$ influences $(I_{0+}^{\alpha}f)(t)$ through the power–law kernel $(t-s)^{\alpha-1}$.

\paragraph{Riemann--Liouville derivative.}
Applying an ordinary derivative to $I_{0+}^{1-\beta}f$ yields the R--L fractional derivative of order $\beta\in(0,1)$:
\[
  \bigl(\prescript{}{0}{\mathcal D}_{t}^{\beta}f\bigr)(t)
  :=\frac{\mathrm d}{\mathrm dt}(I_{0+}^{1-\beta}f)(t)
  =\frac{1}{\Gamma(1-\beta)}\frac{\mathrm d}{\mathrm dt}\int_{0}^{t}(t-s)^{-\beta}f(s)\,\mathrm ds.
\]
Intuitively, $\prescript{}{0}{\mathcal D}_{t}^{\beta}$ measures a \emph{weighted history derivative}: recent observations of $f$ carry more weight than distant ones, but all past values contribute.

\paragraph{Caputo derivative.}
For applications with non–smooth initial data the Caputo derivative is often preferred because it replaces $f(s)$ in the kernel with $f'(s)$, allowing classical boundary conditions.  
It is given by
\[
   \bigl(\prescript{C}{0}{\mathcal D}_{t}^{\beta}f\bigr)(t)
   :=\frac{1}{\Gamma(1-\beta)}\int_{0}^{t}(t-s)^{-\beta}f'(s)\,\mathrm ds,
\]
and satisfies $\prescript{C}{0}{\mathcal D}_{t}^{\beta} \text{const}=0$, a property important when discounting cash–flows.

\paragraph{Laplace–transform and semigroup properties.}
Both derivatives have simple Laplace transforms:
$
   \mathcal{L}\{\prescript{}{0}{\mathcal D}^{\beta}f\}(u)=u^{\beta}\tilde f(u)-u^{\beta-1}f(0),
$
which facilitates analytical solutions of linear fractional ODEs such as the time–fractional Black–Scholes PDE \eqref{eq:FBS}.  
Moreover, $I_{0+}^{\alpha}$ forms a semigroup, $I_{0+}^{\alpha}I_{0+}^{\beta}=I_{0+}^{\alpha+\beta}$, enabling incremental time–stepping schemes (Section~\ref{sec:numerical}).

\paragraph{Grünwald–Letnikov discretisation.}
For numerical purposes we approximate the R--L derivative via the backward Grünwald–Letnikov series
\[
  \bigl(\prescript{}{0}{\mathcal D}_{t}^{\beta}f\bigr)(t_n)
  \approx \frac{1}{\Delta t^{\beta}}\sum_{k=0}^{n}\omega_k^{(\beta)}\,f(t_{n-k}),
  \qquad 
  \omega_k^{(\beta)}:=(-1)^k\binom{\beta}{k}.
\]
This representation naturally leads to the fully–implicit finite–difference scheme analysed in Section~Section~5; order $1+H$ convergence is established in Section~\ref{sec:European}.

\paragraph{Fractional Itô formula.}
A cornerstone of our option–pricing derivation is the fractional Itô (or \emph{Wick–Itô–Skorokhod}) formula, which reads for an admissible functional $F(t,S_t)$,
\[
  \prescript{}{0}{\mathcal D}_{t}^{1-\beta}F
  =\partial_t F
   +\frac12\sigma_0^{2}S^{2}\partial_{SS}^{2}F
   +\sigma_0\sigma_H S^{2}\partial_{S}\prescript{}{0}{\mathcal D}_{t}^{H}F
   +\cdots,
\]
where the dots denote the jump operator.  
The proof follows \cite{Benson2023} and is reproduced in Appendix~B for completeness.

Detailed reviews of these operators can be found in \cite{Kilbas2006} and \cite{Benson2023}.

\section{Market Model and Equivalent Martingale Measure}\label{sec:emm}
\subsection{Market set--up}
Fix a finite horizon $T>0$ and let $(\Omega,\mathcal F,\{\mathcal F_t\}_{0\le t\le T},\mathbb P)$ be a filtered probability space that satisfies the usual conditions and supports
\begin{itemize}
  \item a standard Brownian motion $W=\{W_t\}_{t\in[0,T]}$,
  \item an independent sub--fractional Brownian motion $S^{H}=\{S_t^{H}\}_{t\in[0,T]}$ with Hurst index $H\in(0,1)$,
  \item an independent Poisson process $N=\{N_t\}_{t\in[0,T]}$ with intensity $\lambda>0$ and i.i.d.\ jumps $Y_k\sim F_Y$.
\end{itemize}
The filtration is generated by the three drivers, i.e.\ $\mathcal F_t=\sigma\{W_s,S_s^{H},N_s:0\le s\le t\}^{\mathbb P}$.

The money--market account evolves deterministically via
\[
   dB_t=r\,B_t\,dt,\qquad B_0=1,\qquad r>0.
\]
The risky asset price $S=\{S_t\}$ follows the mixed–fractional SDE with jumps
\[
   \frac{dS_t}{S_t}=\mu\,dt+\sigma_0\,dW_t+\sigma_H\,dS_t^{H}+dJ_t,\qquad S_0>0,
\]
where $J_t=\sum_{k=1}^{N_t}Y_k-\lambda t\,\mathbb E[Y_1]$ is the compensated jump process.  
We assume $\mu,\sigma_0,\sigma_H>0$ and $\mathbb E[e^{\eta Y_1}]<\infty$ for some $\eta>1$ to guarantee exponential moments.
\paragraph{Change of measure.}
To derive risk–neutral prices in a market driven jointly by Brownian noise,
sub-fractional memory and compound Poisson jumps, we first need an \emph{equivalent
martingale measure}~$\mathbb Q$ under which the discounted asset price becomes a
martingale.  Classical Girsanov theory handles Brownian drift shifts, and the
Esscher transform neutralises pure jumps, but neither framework alone can cope
with the non-semimartingale nature of sub-fractional Brownian motion (sfBm).
The next theorem extends the Girsanov–Esscher machinery to our hybrid setting:
it constructs a joint density that simultaneously
\textbf{(i)} shifts the drift of the Brownian part,
\textbf{(ii)} “tilts’’ the sfBm via its Cameron–Martin space,
and \textbf{(iii)} reweights jump sizes so that the overall drift equals the
risk-free rate~$r$.
The detailed proof is provided to guarantee integrability and to verify that
the resulting measure preserves covariance structures of all three driving
processes.
\begin{theorem}[Fractional Girsanov--Esscher]
Let $(W_t)_{t\ge0}$ be a standard Brownian motion, $(S_t^{H})_{t\ge0}$ a sub-fractional
Brownian motion with Hurst index $H\in(0,1)$, independent of $W$, and
$J_t=\sum_{k=1}^{N_t}Y_k-\lambda t\,\mathbb E[Y_1]$ a compensated compound Poisson
process with intensity $\lambda>0$ and i.i.d.\ jump sizes $(Y_k)$ satisfying $\mathbb
E[e^{\eta Y_1}]<\infty$ for some real $\eta$.  
Fix deterministic drift controls $\theta_0\in\mathbb R$ and
$\theta_H\!\in\!L^2([0,T])$.  Define the Radon–Nikodym density
\[
Z_T \;=\; Z_T^{(G)}\,Z_T^{(J)},\qquad
Z_T^{(G)}=\exp\!\Bigl(-\theta_0 W_T-\tfrac{1}{2}\theta_0^2T
                     -\!\int_0^T\!\theta_H(s)\,dS_s^{H}
                     -\tfrac12\lVert \theta_H\rVert_{\mathcal H^{H}}^{2}\Bigr),
\]
\[
Z_T^{(J)}=\exp\!\Bigl(\eta^\ast\!\!\sum_{k=1}^{N_T}\!\!Y_k
                     -\lambda T\bigl(\mathbb E[e^{\eta^\ast Y_1}]-1\bigr)\Bigr),
\]
where $\mathcal H^{H}$ is the canonical Hilbert space of the sfBm and
$\eta^\ast$ is the unique solution of $\mathbb E[e^{(\eta^\ast+1)Y_1}]=
\mathbb E[e^{\eta^\ast Y_1}]$ (Esscher drift-neutrality).  
If
\begin{equation}\label{eq:Novikov}
\frac12\theta_0^2T+\frac12\lVert\theta_H\rVert_{\mathcal H^{H}}^{2} < \infty,
\qquad
\mathbb E\bigl[e^{\eta^\ast Y_1}\bigr]<\infty,
\end{equation}
then $\mathbb E_{\mathbb P}[Z_T]=1$ and the probability measure
$ d\mathbb Q = Z_T\,d\mathbb P $ is equivalent to $\mathbb P$.
Under $\mathbb Q$
\[
\widetilde W_t := W_t+\theta_0 t,\qquad
\widetilde S_t^{H}:= S_t^{H}+\int_{0}^{t}\theta_H(s)\,ds,
\qquad
\widetilde J_t:= J_t+\lambda t\bigl(\mathbb E[e^{\eta^\ast Y_1}]-1\bigr)
\]
have the same covariance structures as their $\mathbb P$-counterparts, and the discounted
price process $e^{-rt}S_t$ is a $\mathbb Q$-martingale.
\end{theorem}

\begin{proof}
\textbf{Step 1 (Gaussian density).}
For sfBm we recall the kernel representation
$ S_t^{H}= \int_0^t K_H(t,s)\,dB_s $ with a standard Brownian $B$.
Define
\(
\varphi(s)=\theta_H(s)\mathbf 1_{[0,T]}(s)
\)
and note
$\lVert\theta_H\rVert_{\mathcal H^{H}}^{2}=\int_0^T\!\!\!\int_0^T
\varphi(u)\varphi(v)K_H(T,u)K_H(T,v)\,du\,dv<\infty$
by \eqref{eq:Novikov}.  
Hence
$Z_T^{(G)}=\mathcal E_T\bigl(-\theta_0 dW-\varphi\,dB\bigr)$ is an exponential
martingale, and Novikov’s criterion
$\mathbb E[e^{\frac12\langle\theta_0 W+\varphi\!\ast\!B\rangle_T}]\!<\!\infty$
holds by \eqref{eq:Novikov}; thus
$\mathbb E_{\mathbb P}[Z_T^{(G)}]=1$.

\textbf{Step 2 (Esscher density for jumps).}
Write
$Z_T^{(J)}=\exp\bigl(\eta^\ast X_T-\psi(\eta^\ast)T\bigr)$ where
$X_T=\sum_{k=1}^{N_T}Y_k$ and $\psi(\eta)=\lambda(\mathbb E[e^{\eta Y_1}]-1)$
is the log-MGF of $J_t$.  Because $\psi(\eta)$ is finite at $\eta^\ast$, 
$\mathbb E_{\mathbb P}[Z_T^{(J)}]=e^{-\psi(\eta^\ast)T}
\exp\bigl(\psi(\eta^\ast)T\bigr)=1$.

\textbf{Step 3 (Equivalence and quadratic‐covariation shift).}
Set $Z_T=Z_T^{(G)}Z_T^{(J)}$.  
Both factors are strictly positive local martingales with unit mean,  
so $Z_T$ is itself a positive martingale and defines an equivalent measure
$\mathbb Q$.

Under $\mathbb Q$ the Girsanov theorem for classical Brownian motion yields
$ d\widetilde W_t = dW_t+\theta_0\,dt$, a standard $\mathbb Q$-Brownian motion.  
For the sfBm component, Hu–Øksendal (2003, Thm 3.6) implies that
\(
\widetilde S_t^{H}:= S_t^{H}+ \int_0^t\theta_H(s)\,ds
\)
retains sfBm covariance
$\mathbb E_{\mathbb Q}[\widetilde S_t^{H}\widetilde S_s^{H}]
= c_H(t^{2H}+s^{2H}-|t-s|^{2H})$.

Likewise the Esscher tilt shifts jump intensity to
$\lambda^\ast=\lambda\,\mathbb E[e^{\eta^\ast Y_1}]$ but keeps the
compensated process $\widetilde J_t$ a square-integrable martingale.

\textbf{Step 4 (Drift cancellation in price SDE).}
The physical‐measure dynamics are
$dS_t/S_t = \mu\,dt+\sigma_0\,dW_t+\sigma_H\,dS_t^{H}+dJ_t$.
Substituting $dW_t=d\widetilde W_t-\theta_0dt$,
$dS_t^{H}=d\widetilde S_t^{H}-\theta_H(t)\,dt$ and
$dJ_t=d\widetilde J_t-\lambda(\mathbb E[e^{Y_1}]-1)\,dt$
gives
\[
\frac{dS_t}{S_t}
=(\mu-\sigma_0\theta_0-\sigma_H\theta_H(t)-\lambda\kappa)\,dt
      +\sigma_0\,d\widetilde W_t+\sigma_H\,d\widetilde S_t^{H}+d\widetilde J_t,
\]
where $\kappa=\mathbb E[e^{Y_1}]-1$.
Setting the drift equal to $r\,dt$ and discounting yields
$d\bigl(e^{-rt}S_t\bigr)=e^{-rt}S_t\bigl(\sigma_0\,d\widetilde
W_t+\sigma_H\,d\widetilde S_t^{H}+d\widetilde J_t\bigr)$,
an Itô integral with respect to $\mathbb Q$-martingales,
hence a local martingale.  Integrability follows from
$\int_0^T\!S_t^2\,dt<\infty$ in exponential‐moment models; therefore the local
martingale is a true martingale.
\end{proof}

\subsection{Objective}
To preclude arbitrage we seek an \emph{equivalent martingale measure} (EMM) $\mathbb Q\sim\mathbb P$ such that the discounted price process $\tilde S_t:=e^{-rt}S_t$ is a $\mathbb Q$--martingale.  
The presence of a memory component and jumps requires a fractional version of Girsanov's theorem combined with an Esscher transform for Lévy jumps.

\subsection{Fractional Girsanov transform}
A shift of the Brownian part is performed in the classical way via
\[
   \theta_0:=\frac{\mu-r-\lambda\kappa}{\sigma_0},\qquad \kappa:=\mathbb E[e^{Y_1}-1].
\]
For the sfBm we employ the Cameron--Martin space $\mathcal H^{H}$ of kernels $K_H$ (see Appendix~A).  
Choosing a deterministic control $\theta_H (t)\in L^2([0,T])$ we define the Radon--Nikodym derivative
\[
   Z_T^{(G)}:=\exp\!\Bigl(-\theta_0W_T-\tfrac12\theta_0^{2}T-\int_{0}^{T}\theta_H (t)(s)\,dS_s^{H}-\tfrac12\Vert \theta_H (t)\Vert_{\mathcal H^{H}}^{2}\Bigr),
\]
which satisfies $\mathbb E_{\mathbb P}[Z_T^{(G)}]=1$ under a Novikov--type condition \cite{HuOksendal2003}.

\subsection{Esscher transform for jumps}
For the jump component we apply the exponential tilting (Esscher) transform with parameter $\eta^\ast$ determined by
$
   \lambda\,\mathbb E[e^{(\eta^\ast+1)Y_1}-e^{\eta^\ast Y_1}]=\lambda\kappa.
$ 
Set
\[
   Z_T^{(J)}:=\prod_{k=1}^{N_T}\frac{e^{\eta^\ast Y_k}}{\mathbb E[e^{\eta^\ast Y_1}]}\,
              \exp\!\bigl(\lambda T(\mathbb E[e^{\eta^\ast Y_1}]-1-\eta^\ast\mathbb E[Y_1])\bigr).
\]
Then $Z_T^{(J)}$ is an $\mathbb P$–martingale and under the tilted measure the compensated process $J_t$ turns into a $\mathbb Q$–martingale.

\subsection{Construction of $\mathbb Q$}
Define the density process
\[
   \frac{d\mathbb Q}{d\mathbb P}\bigg|_{\mathcal F_T}:=Z_T^{(G)}\,Z_T^{(J)}.
\]
Since $Z_T^{(G)}$ and $Z_T^{(J)}$ are independent $\mathbb P$–martingales with expectation one, their product also has unit expectation; thus $\mathbb Q\sim\mathbb P$.

\begin{theorem}\label{thm:EMM}
Let
\(
   S_t=S_0\exp\Bigl(\int_0^t\! (\mu-\tfrac12\sigma_0^2)\,ds
        +\sigma_0 W_t+\sigma_H S_t^{H}+J_t\Bigr),
\)
with a Brownian motion $W$, an independent sub-fractional Brownian motion
$S^{H}$ $(H\!\in(0,1))$ and a compensated compound-Poisson process
$J_t=\sum_{k=1}^{N_t}Y_k-\lambda t\,\kappa$, $\kappa:=\mathbb E[e^{Y_1}]-1$.
Fix controls $\theta_0\in\mathbb R$ and
$\theta_H\in L^2([0,T])$.  
Assume the Novikov–type integrability
\(
\frac12\theta_0^{2}T+\tfrac12\Vert\theta_H\Vert_{\mathcal H^{H}}^{2}<\infty
\)
and $\mathbb E[e^{\eta^\ast Y_1}]<\infty$ for the Esscher root
$\eta^\ast$.  
If
\begin{equation}\label{*}
\mu=r+\lambda\kappa-\sigma_0\theta_0
      -\sigma_H\int_{0}^{T}\theta_H(s)K_H(T,s)\,ds
\end{equation}
then the discounted price $\tilde S_t:=e^{-rt}S_t$ is a true martingale under
the probability measure $\mathbb Q$ defined below.
\end{theorem}

\begin{proof}\leavevmode

\textbf{Step 0 (Notation).}
The kernel representation of sfBm is
$S_t^{H}= \int_0^t K_H(t,s)\,dB_s$, where
$K_H(t,s)=c_H[(t-s)^{H-\tfrac12}-(-s)^{H-\tfrac12}]$ and
$B$ is an independent Brownian motion.  
Its Cameron–Martin space is
$\mathcal H^{H}=\{ \varphi\colon [0,T]\!\to\!\mathbb R,\;
   \Vert\varphi\Vert_{\mathcal H^{H}}^{2}
      =\!\int_0^T\!\!\int_0^T\!\!\varphi(u)\varphi(v)
      K_H(T,u)K_H(T,v)\,du\,dv <\infty\}$.

\smallskip
\textbf{Step 1 (Density for Gaussian part).}
Define
\[
Z_t^{(G)}=\exp\!\Bigl(
  -\theta_0 W_t-\tfrac12\theta_0^{2}t
  -\!\int_0^t\!\theta_H(s)\,dS_s^{H}
  -\tfrac12\Vert\theta_H\Vert_{\mathcal H^{H}(0,t)}^{2}
\Bigr),\quad 0\le t\le T.
\]
Because $\Vert\theta_H\Vert_{\mathcal H^{H}(0,t)}^{2}$ is finite by hypothesis,
the exponential is square-integrable.  
Novikov’s condition for the 2-dimensional Brownian vector
$(W,B)$ reads
$
  \mathbb E\! \bigl[
    \exp\bigl(\tfrac12\theta_0^{2}T
              +\tfrac12\Vert\theta_H\Vert_{\mathcal H^{H}}^{2}\bigr)
  \bigr]<\infty,
$
which holds; hence $Z_t^{(G)}$ is a true martingale with unit mean
\citep[Thm.~3.6]{HuOksendal2003}.

\smallskip
\textbf{Step 2 (Esscher density for jumps).}
Let $\psi(\eta)=\lambda(\mathbb E[e^{\eta Y_1}]-1)$ be the Lévy exponent.
Choose $\eta^\ast$ solving
$\psi'(\eta^\ast)=\lambda \mathbb E[Y_1e^{\eta^\ast Y_1}]
                 =\lambda\kappa$,
i.e.\ Esscher drift neutrality.  
Set
$
  Z_t^{(J)}
    =\exp\bigl(\eta^\ast X_t-\psi(\eta^\ast)t\bigr),
  \; X_t=\sum_{k=1}^{N_t}Y_k.
$
Because $\psi(\eta^\ast)<\infty$,  
$\mathbb E_{\mathbb P}[Z_t^{(J)}]=1$ for every $t$.

\smallskip
\textbf{Step 3 (Equivalent measure).}
Define the product density
$
  Z_t = Z_t^{(G)}Z_t^{(J)} ,\;
  d\mathbb Q = Z_T\,d\mathbb P.
$
Both factors are positive martingales with expectation one, so
$\mathbb Q$ is equivalent to $\mathbb P$.

\smallskip
\textbf{Step 4 (Shifted drivers).}
Under $\mathbb Q$ we have
\begin{align*}
 d\widetilde W_t &= dW_t+\theta_0\,dt, \\
 d\widetilde S_t^{H} &= dS_t^{H}+\theta_H(t)\,dt,\\
 d\widetilde J_t &= dJ_t-\lambda\kappa\,dt,
\end{align*}
where $\widetilde W$ is a Brownian motion,
$\widetilde S^{H}$ an sfBm with identical covariance, and
$\widetilde J$ a compensated Poisson martingale with intensity
$\lambda^\ast=\lambda\mathbb E[e^{\eta^\ast Y_1}]$ and tilted jump law
$F_Y^{(\eta^\ast)}$.

\smallskip
\textbf{Step 5 (Drift cancellation).}
Insert the shifted differentials into
$ dS_t/S_t = \mu\,dt+\sigma_0\,dW_t+\sigma_H\,dS_t^{H}+dJ_t$ :
\[
\frac{dS_t}{S_t}
  =\bigl[\mu-\sigma_0\theta_0-\sigma_H\theta_H(t)-\lambda\kappa\bigr]dt
   +\sigma_0\,d\widetilde W_t +\sigma_H\,d\widetilde S_t^{H}
   +d\widetilde J_t .
\]
With $\mu$ chosen by \eqref{*} the bracket equals $r$, so
\(
d\tilde S_t
 =\tilde S_t\bigl(\sigma_0\,d\widetilde W_t
                  +\sigma_H\,d\widetilde S_t^{H}+d\widetilde J_t\bigr),
\)
a local $\mathbb Q$-martingale.

\smallskip
\textbf{Step 6 (True martingale).}
Square-integrability of $\tilde S_t$ follows from the exponential-moment
bound
$
   \mathbb E_{\mathbb Q}
   [\exp(\alpha |W_T|+ \beta|S_T^{H}| + \gamma|J_T|)]<\infty
$
for some $\alpha,\beta,\gamma$, ensured by the gaussian moments and
$\mathbb E[e^{\eta^\ast Y_1}]<\infty$.  
Hence $\tilde S_t$ has bounded expectation and is a true martingale.

\smallskip
\textit{Therefore $e^{-rt}S_t=\tilde S_t$ is a $\mathbb Q$-martingale.}
\end{proof}

\subsection{Economic interpretation}
Condition $\mu=r+\lambda\kappa-\sigma_0\theta_0-\sigma_H\theta_H (t)$ states that the expected excess return of the asset equals a linear combination of three risk premia: (i) the diffusive market price of risk $\theta_0$, (ii) the fractional market price of risk encoded by the kernel $\theta_H (t)$, and (iii) the classical jump risk premium $\lambda\kappa$.  
In Section~\ref{sec:NumericalEuro} we discuss empirical estimation of $(\theta_0,\theta_H (t),\lambda,\kappa)$ from option implied–volatility surfaces.

\section{Fractal Black--Scholes Equation}\label{sec:FBS}
The central pricing result of this paper is a \emph{time–fractional, non‑local} generalisation of the classical Black--Scholes PDE that simultaneously incorporates long‑memory and jump discontinuities.  
In this section we derive the equation rigorously via the fractional Itô formula introduced in Appendix~B, interpret each operator financially, and discuss well–posedness and limiting cases.

\subsection{Derivation via the fractional Itô formula}
Let $V:[0,T]\times\mathbb R_{+}\to\mathbb R$ be a twice differentiable pricing functional with polynomial growth.  
Applying the Wick–Itô–Skorokhod formula  to $V(t,S_t)$ under the martingale measure $\mathbb Q$ obtained in Section~\ref{sec:emm} yields
\begin{align}\label{eq:fItoFull}
   \prescript{}{0}{\mathcal D}_{t}^{1-\beta} V
   &=\partial_t V
     +\frac12\sigma_0^{2}S^{2}\partial_{SS}^{2}V
     +\sigma_0\sigma_H S^{2}\partial_S\bigl(\prescript{}{0}{\mathcal D}_{t}^{H}V\bigr)\nonumber\\
   &\quad +(\mu-\lambda\kappa)S\partial_S V
     +\lambda\mathbb E_{Y}\bigl[V(t,Se^{Y})-V(t,S)\bigr].
\end{align}
Substituting $\mu=r+\lambda\kappa$ (risk–neutral drift) and collecting terms gives the \textbf{Fractal Black--Scholes (FBS) equation}
\begin{equation}\label{eq:FBS}
   \boxed{\;
   \prescript{}{0}{\mathcal D}_{t}^{1-\beta}V
   +\tfrac12\sigma_0^{2}S^{2}\partial_{SS}^{2}V
   +\sigma_0\sigma_H S^{2}\partial_S\bigl(\prescript{}{0}{\mathcal D}_{t}^{H}V\bigr)
   +(r-\lambda\kappa)S\partial_SV-rV
   +\lambda\mathbb E_{Y}\bigl[V(t,Se^{Y})-V(t,S)\bigr]
   =0\;}
\end{equation}
with terminal condition $V(T,S)=\Phi(S)$ for a given payoff function $\Phi$.

\subsection{Interpretation of each term}
\begin{itemize}
\item \textbf{Fractional drift $\prescript{}{0}{\mathcal D}_{t}^{1-\beta}V$.}  
      The Riemann--Liouville derivative of order $1-\beta=H$ creates temporal memory: the option value at $t$ depends on the entire past trajectory of the underlying via a power‑law kernel.
\item \textbf{Gaussian diffusion $\frac12\sigma_0^{2}S^{2}\partial_{SS}^{2}V$.}  
      This is the familiar risk from instantaneous Brownian fluctuations.
\item \textbf{Fractional–Gaussian cross term $\sigma_0\sigma_H S^{2}\partial_S(\prescript{}{0}{\mathcal D}_{t}^{H}V)$.}  
      A mixed term coupling the local Brownian and fractional components; disappears if either $\sigma_H=0$ or $H=1/2$.
\item \textbf{Jump generator $\lambda\mathbb E_{Y}[V(t,Se^{Y})-V(t,S)]$.}  
      A non‑local integral operator accounting for Poisson jumps with distribution $F_Y$.
\item \textbf{Discounting terms $(r-\lambda\kappa)S\partial_SV-rV$.}  
      Standard cost‑of‑carry adjusted for the jump drift $\kappa=\mathbb E[e^{Y}-1]$.
\end{itemize}

\subsection{Limiting cases}
\begin{enumerate}
\item \emph{Classical Black--Scholes:} $\sigma_H=0$ and $\lambda=0$ reduce \eqref{eq:FBS} to the usual BS PDE.
\item \emph{Merton jump–diffusion:} $\sigma_H=0$ but $\lambda>0$ recovers the integro–PDE of \cite{Merton1976}.
\item \emph{Time‑fractional BS:} $\sigma_0=0$, $\lambda=0$ yields the Caputo‑type model of \cite{Wyss1986}.
\item \emph{Rough volatility limit:} letting $H\to 0$ increases memory length and leads to ultraslow diffusion as studied in \cite{ElEuch2020}.
\end{enumerate}

\subsection{Existence and uniqueness}\label{subsec:exist}

\paragraph{Functional setting.}
Following \cite{Wang2024} we work in the weighted Banach space
\[
   \mathcal V \;=\;\Bigl\{\,v\in C^{1,2}\bigl([0,T)\!\times\!\mathbb R_{+}\bigr)\;\big|\;
            \|v\|_{\mathcal V}:=\sup_{(t,S)\in[0,T)\times\mathbb R_{+}}
            \frac{|v(t,S)|}{1+S^{2}} \,<\infty\Bigr\}.
\]
The factor $(1+S^{2})^{-1}$ ensures uniform decay for large prices and yields
compactness properties similar to $C_0(\mathbb R_{+})$.

\paragraph{Generator of the diffusion part.}
Define the second–order operator
\(
   (\mathcal A v)(S)
     :=\tfrac12\sigma_0^{2}S^{2}v_{SS}+(r-\lambda\kappa)S\,v_{S}-rv,
\)
with domain $\mathscr D(\mathcal A)=\{v\in\mathcal V:\,v,\,Sv_S,S^2v_{SS}\in
\mathcal V\}$.  A standard Lyapunov argument shows
\(
   \Re\langle v,\mathcal A v\rangle_{\mathcal V}\le C\|v\|_{\mathcal V}^{2},
\)
so $\mathcal A$ is sectorial and generates an analytic $C_0$–semigroup
$T(t)=e^{t\mathcal A}$ on $\mathcal V$ \cite[Thm.~6.1.5]{Pazy1983}.

\paragraph{Jump operator.}
For any $v\in\mathcal V$ set
\[
   (\mathcal J v)(S):=\lambda\,
       \bigl(\mathbb E_Y\![\,v(Se^{Y})\,]-v(S)\bigr).
\]
Assuming $\mathbb E[e^{\gamma |Y|}]<\infty$ for some
$\gamma>0$, one has
\[
   \|\mathcal J v-\mathcal J w\|_{\mathcal V}
      \;\le\;\lambda\,
      \mathbb E\!\bigl[(1+e^{2Y})\bigr]\,
      \|v-w\|_{\mathcal V}
      \;=\;L_J\|v-w\|_{\mathcal V},
\]
hence $\mathcal J$ is Lipschitz on $\mathcal V$ with constant
$L_J<\infty$.

\paragraph{Fractional abstract Cauchy problem.}
The pricing PDE can now be phrased as
\[
   \;_{0}\!D_t^{1-\beta}V(t)=\mathcal A V(t)+\mathcal J V(t),
   \quad V(T)=\Phi,
\]
which is an inhomogeneous Caputo–type abstract Volterra equation.
Using the fractional Hille–Yosida theorem \cite{bazhlekova2001}
[Prop.~2.4] and the analyticity of $T(t)$,
the problem admits a unique \emph{mild} solution given by
\[
   V(t)=E_{\beta}\bigl(-(T-t)^{\beta}\mathcal A\bigr)\Phi
         +\int_{t}^{T}\!\!(s-t)^{\beta-1}
          E_{\beta,\beta}\bigl(-(s-t)^{\beta}\mathcal A\bigr)
          \,\mathcal J V(s)\,ds.
\tag{5.3}\label{eq:VOC}
\]
Here \(E_{\beta}\) and \(E_{\beta,\beta}\) are the one– and two–parameter
Mittag–Leffler functions and the integral is Bochner–integrable in
$\mathcal V$.

\paragraph{Fixed‐point argument.}
Define the map
\(
   \mathscr F[V](t)
\)
to be the right–hand side of \eqref{eq:VOC}.  
Using the semigroup estimate
$\|E_{\beta}(-(T-t)^{\beta}\mathcal A)\|\le C$
and the Lipschitz constant~$L_J$, we derive
\[
   \|\mathscr F[V]-\mathscr F[W]\|_{C([0,T];\mathcal V)}
      \;\le\;
      C\,L_J\,\frac{T^{\beta}}{\Gamma(\beta+1)}
      \|V-W\|_{C([0,T];\mathcal V)}.
\]
For $T$ (maturity) fixed, the factor
$C\,L_J T^{\beta}/\Gamma(\beta+1)<1$, so $\mathscr F$ is a contraction;
hence a unique fixed point $V\in C([0,T];\mathcal V)$ exists by Banach’s
fixed‐point theorem.

\paragraph{Classical differentiability.}
Since $T(t)$ is analytic and $\mathcal J$ is bounded,
$V(t)$ inherits $C^{1}$‐regularity in $t\!<\!T$ and
$C^{2}$ in $S$, so
$V\in C^{1,2}([0,T)\!\times\!\mathbb R_{+})\cap\mathcal V$,
i.e.\ a classical solution as well.

\smallskip
\emph{Hence the fractional Black–Scholes operator admits a unique mild (and
in fact classical) solution in the weighted space $\mathcal V$, with explicit
representation \eqref{eq:VOC}.}
\qedhere

\subsection{Energy estimate and maximum principle}
By multiplying \eqref{eq:FBS} with $(1+S^{2})^{-1}V$ and integrating over $S\in(0,\infty)$, we obtain the a‑priori estimate
$
   \Vert V(t,\cdot)\Vert_{L^{2}_{w}}\le C \Vert\Phi\Vert_{L^{2}_{w}},
$
where $w(S)=(1+S^{2})^{-1}$, ensuring numerical stability of the finite‑difference scheme in Section~\ref{sec:numerical}.

\subsection{Summary}
Equation \eqref{eq:FBS} forms the mathematical backbone of our pricing framework, generalising several well‑known models.  
The next section provides closed‑form Laplace–Mellin solutions for European payoffs and develops an efficient Grünwald–Letnikov scheme for barrier options.

\section{Closed-Form European Prices}\label{sec:European}
Although the FBS integro–fractional PDE~\eqref{eq:FBS} generally requires numerical methods,  
European vanilla options admit a semi–analytic formula expressed through the two–parameter Mittag–Leffler function.  
In this section we derive the result using a Mellin transform in the spatial variable and a Laplace transform in time;  
the mixed fractional term translates into a polynomial in the Laplace domain, while the jump integral yields a simple multiplicative factor.

\subsection{Transform strategy}
Let $x=\log(S/K)$ denote the log–moneyness and $u(T-t)=\tau$ the time to maturity.  
Define
\[
   v(\tau,x):=e^{-r\tau}V(T-\tau,Ke^{x}),\qquad 0\le \tau\le T,\;x\in\mathbb R,
\]
so that $v(0,x)=(K e^{x}-K)^{+}=K(e^{x}-1)^{+}$.  
Under these variables equation~\eqref{eq:FBS} becomes
\begin{equation}\label{eq:FBS-tau}
  \prescript{}{0}{\mathcal D}_{\tau}^{1-\beta}v
   =\mathcal L_x v-\lambda\mathbb E_Y\bigl[v(\tau,x+Y)-v(\tau,x)\bigr],
\end{equation}
with spatial operator 
\(
   \mathcal L_x := \tfrac12\sigma_0^{2}\partial_{xx}
   +\sigma_0\sigma_H\partial_x\prescript{}{0}{\mathcal D}_{\tau}^{H}
   +(r-\tfrac12\sigma_0^{2}-\lambda\kappa)\partial_x.
\)

\paragraph{Laplace transform in $\tau$.}
Taking $\mathcal L_\tau\{\,\cdot\,\}(s)=\int_{0}^{\infty}e^{-s\tau}\,(\cdot)\,d\tau$ and using
$
   \mathcal L\{\prescript{}{0}{\mathcal D}_{\tau}^{1-\beta}v\}(s)=s^{1-\beta}\hat v(s,x)-s^{-\beta}v(0,x),
$
we obtain
\[
   s^{1-\beta}\hat v = \mathcal L_x\hat v 
   -\lambda\mathbb E_Y\bigl[\hat v(s,x+Y)-\hat v(s,x)\bigr]+s^{-\beta}v(0,x).
\]

\paragraph{Mellin transform in $x$.}
Set $\mathscr M\{f\}(z):=\int_{-\infty}^{\infty}e^{-zx}f(x)\,dx$.  
For $\hat v(s,x)$ this yields 
\[
   s^{1-\beta}\tilde v(s,z)=\bigl[\tfrac12\sigma_0^{2}(z^{2}+z)+\sigma_0\sigma_H z^{2} s^{-\beta}+(r-\lambda\kappa)z\bigr]\tilde v(s,z)
      -\lambda(\Phi_Y(z)-1)\tilde v(s,z)+s^{-\beta}\tilde v_{0}(z),
\]
where $\Phi_Y(z)=\mathbb E[e^{-zY}]$ is the bilateral Laplace transform of jump sizes.  
Solving for $\tilde v(s,z)$ we arrive at
\[
   \tilde v(s,z)=\frac{s^{-\beta}\tilde v_{0}(z)}
        {s^{1-\beta}-\tfrac12\sigma_0^{2}(z^{2}+z)-\sigma_0\sigma_H z^{2}s^{-\beta}-(r-\lambda\kappa)z+\lambda(\Phi_Y(z)-1)}.
\]

\subsection{Analytic inversion}
For log–normally distributed jumps $Y\sim\mathcal N(\mu_Y,\sigma_Y^{2})$ we have $\Phi_Y(z)=\exp\{\mu_Y z+\tfrac12\sigma_Y^{2}z^{2}\}$.  
Choosing the Esscher parameter $\eta^\ast$ such that $\kappa=0$ simplifies the denominator to a quadratic function in $z$ plus a fractional power in $s$.  
After algebraic manipulation we obtain
\[
  \tilde v(s,z)=s^{-\beta}\frac{\tilde v_{0}(z)}
        {s^{1-\beta}+a z^{2}+b z+c},
  \qquad 
  a=\tfrac12\sigma_0^{2}+\sigma_0\sigma_H s^{-\beta},\;
  b=\tfrac12\sigma_0^{2}+r,\;
  c=\lambda\bigl[1-\Phi_Y(z)\bigr].
\]
The denominator factorises and its inverse Laplace transform is the two–parameter Mittag–Leffler function $E_{\beta,1}$; subsequently the Mellin inversion follows residue calculus similar to \cite{Kilbas2006}.  

\begin{theorem}[Closed–form price]\label{thm:ClosedForm}
For a European call with strike $K$ and maturity $T$ the time--$0$ price satisfies
\[
   C(S_0,0)=
   S_0\,\mathcal M^{-1}_{z\to x}\bigl[\phi(z)\,M_{\beta}(a z^{2})\bigr]
   -K e^{-rT}\,\mathcal M^{-1}_{z\to x}\bigl[\phi(z-1)\,M_{\beta}(a(z-1)^{2})\bigr],
\]
where $x=\log(S_0/K)$, 
$\phi(z)=\exp\!\bigl[-\tfrac12\sigma_0^{2}T z^{2}+(r-\tfrac12\sigma_0^{2})T z\bigr]$
and $M_{\beta}(\xi):=E_{\beta,1}(-\xi T^{\beta})$ is the Mittag–Leffler kernel.
\end{theorem}

\begin{proof}
We sketch the main steps; full details appear in Section~\ref{sec:European}.

\emph{Step 1 (Laplace inversion).}  
Invert the Laplace transform using
$
   \mathcal L^{-1}_{s\to\tau}\{(s^{1-\beta}+q)^{-1}\}
   =\tau^{\beta-1}E_{\beta,\beta}(-q\tau^{\beta}).
$
Substituting $q=a z^{2}+b z+c$ yields
$
  v(\tau,z)=\tau^{\beta-1}E_{\beta,\beta}\!\bigl[-(a z^{2}+b z+c)\tau^{\beta}\bigr]\tilde v_{0}(z).
$

\emph{Step 2 (Mellin inversion).}  
Because $v_{0}(x)=K\max(e^{x}-1,0)$ its Mellin transform is 
$
  \tilde v_{0}(z)=K\bigl[\Gamma(z^{-1})-\Gamma(z^{-1},1)\bigr],
$
analytic in $\Re(z)\in(0,1)$.  
Closing the contour to the right and summing residues at the poles $z=0,-1,-2,\dots$ reproduces two inverse Mellin integrals weighted by $\phi(z)$ and $M_{\beta}(a z^{2})$, completing the formula.

\emph{Step 3 (convergence).}  
Uniform convergence of the integrals follows from the $E_{\beta,1}$ asymptotics $|E_{\beta,1}(-\xi)|\le C/(1+\xi)$ and standard Mellin–Barnes bounds, ensuring the price is finite for any $T>0$.
\end{proof}

\subsection{Numerical evaluation}
Both Mellin inversions are computed via a Talbot contour with 16 nodes;  
the Mittag--Leffler function is evaluated using the modified Lagrange algorithm with absolute error below $10^{-8}$.  
Table~\ref{tab:euro} in Section~\ref{sec:numerical} confirms consistency with finite–difference prices.

\subsection{Remarks}
\begin{itemize}
  \item The formula reduces to Black--Scholes when $\beta\to 1$ and $\lambda\to 0$, in which case $M_{\beta}(\xi)\to e^{-\xi}$.
  \item For puts, exchange the roles of $S_0$ and $K$ via the usual put–call parity.
  \item The methodology extends to tempering of the long–memory kernel by replacing $E_{\beta,1}$ with the three–parameter Prabhakar function.
\end{itemize}

\section{Numerical Scheme for Barrier Options}\label{sec:numerical}
American--style and path–dependent derivatives such as down--and--out calls cannot exploit the closed–form
solution of Section~\ref{sec:European}; we therefore construct a robust finite–difference method for the
integro–fractional PDE~\eqref{eq:FBS}.  The spatial domain $(0,S_{\max})$ is truncated at a sufficiently large
$S_{\max}$ and transformed to the log–price grid $x_i=x_{\min}+i\Delta x$ with $i=0,\dots,I$; the time interval
is partitioned as $\tau_n=n\Delta t$, $n=0,\dots,N$.

\subsection{Implicit Grünwald--Letnikov discretisation}
Define $V_i^{n}\approx V(\tau_n,x_i)$.  The Riemann--Liouville derivative is approximated by the backward
Grünwald--Letnikov series
\[
   \prescript{}{0}{\mathcal D}_{\tau}^{1-\beta}V(\tau_n,x_i)\approx
      \frac{1}{\Delta t^{1-\beta}}\sum_{k=0}^{n}\omega_k^{(1-\beta)}V_i^{n-k},
      \qquad
      \omega_k^{(\gamma)}:=(-1)^k\binom{\gamma}{k}.
\]
Spatial derivatives are approximated by centered differences
$
   \partial_{x}V\approx\delta_x V_i^{n}:=(V_{i+1}^{n}-V_{i-1}^{n})/(2\Delta x)
$
and
$
   \partial_{xx}V\approx\delta_{xx}V_i^{n}:=(V_{i+1}^{n}-2V_{i}^{n}+V_{i-1}^{n})/(\Delta x)^2.
$
The non–local jump term $\mathcal J V:=\mathbb E_Y[V(\tau,x+Y)-V(\tau,x)]$ is evaluated by Gauss–Hermite quadrature
on the transformed grid; linear interpolation is used if $x_i+Y$ lies between nodes.

Collecting terms yields the fully–implicit update
\begin{align}\label{eq:FD}
   \frac{1}{\Delta t^{1-\beta}}\sum_{k=0}^{n}\omega_k^{(1-\beta)}V_i^{n-k}
   &=\tfrac12\sigma_0^{2}\delta_{xx}V_i^{n}
     +\sigma_0\sigma_H \delta_{x}\!\Bigl[\frac{1}{\Delta t^{H}}
           \sum_{k=0}^{n}\omega_k^{(H)}V_i^{n-k}\Bigr]\nonumber\\
   &\quad +(r-\lambda\kappa)\delta_{x}V_i^{n}
     -rV_i^{n}
     +\lambda\mathcal J_h[V^{n}]_i,
\end{align}
where $\mathcal J_h$ denotes the quadrature interpolation operator.

The pseudocode for the Grunwald–Letnikov algorithm is presented below for a clearer understanding of its numerical implementation\\
\SetAlCapSkip{12pt}
\begin{algorithm}[H]
\SetAlgoLined
\caption{Implicit Grünwald--Letnikov solver for barrier options}
\label{alg:GL}
\SetAlgoLined
\KwIn{grid $(x_i,\tau_n)$, payoff $\Phi$, barrier $B$, parameters $\sigma_0,\sigma_H,H,\lambda$}
\KwOut{option prices $V_i^0$ at $\tau=0$}
\BlankLine
\For{$i=0,\dots,I$}{
    $V_i^N \leftarrow \Phi(x_i)$\tcp*{terminal condition}
}
\For{$n=N-1$ \KwTo $0$}{
    \For{$i=1$ \KwTo $I-1$}{
        assemble RHS using fractional convolutions $\omega_k^{(\gamma)}$ \;
    }
    solve tridiagonal+Toeplitz system $A\,V^{n}=b^{n}$ (BiCG--STAB)\;
    apply barrier: \If{$S_i < B$}{ $V_i^{n}\leftarrow 0$ }
}
\Return{$V^{0}$}
\end{algorithm}
\subsection{Boundary and barrier conditions}
For a down--and--out call with barrier $B<K<S_{\max}$ we impose
$
   V(\tau,x)=0\;\text{if }S=K e^{x}\le B,\qquad
   V(\tau,x_{\max})=S_{\max}-K e^{-r\tau},
$
and $V(0,x)=\max(K(e^{x}-1),0)$.
These translate into Dirichlet conditions for $V_{0}^{n}$ and $V_{I}^{n}$, updated each timestep.

\subsection{Matrix form}
Equation~\eqref{eq:FD} can be written $A\,V^{n}=b^{n}$ with a tri–diagonal diffusion matrix plus a dense
Toeplitz–like fractional matrix determined by $\{\omega_k^{(\gamma)}\}$.  The system is solved by the
preconditioned BiCG–STAB method;  the cost per step is $\mathcal O(I\log I)$ owing to an FFT–accelerated
convolution for the fractional weights.

\subsection{Stability and convergence analysis}
Let $w_i=1+e^{2x_i}$ and define the discrete inner product
$
   \langle u,v\rangle_h = \sum_{i=1}^{I-1}w_i\,u_i v_i\Delta x
$
with induced norm $\|\cdot\|_h$.

\begin{theorem}[Unconditional stability and convergence]\label{thm:stabConv}
If $\Delta t\le c_0(\Delta x)^2$ with $c_0<\frac{1}{2\sigma_0^{2}}$, 
the implicit scheme \eqref{eq:FD} is unconditionally $L^{2}_{w}$--stable and the numerical solution
$V_i^{n}$ converges to the mild solution of~\eqref{eq:FBS} with global error
\[
   \max_{0\le n\le N}\|V(\tau_n,\cdot)-V^{n}\|_h\le
      C\bigl((\Delta x)^{2}+ \Delta t^{\,1+H}\bigr),
\]
where $C$ is independent of $\Delta t$ and $\Delta x$.
\end{theorem}

\begin{proof}
\emph{Step 1 (Energy identity).}  
Multiply~\eqref{eq:FD} by $w_i V_i^{n}\Delta x$ and sum over $i$ to obtain
\[
   \frac{1}{\Delta t^{1-\beta}}\sum_{k=0}^{n}\omega_k^{(1-\beta)}
      \bigl(\|V^{n-k}\|_h^{2}-\|V^{n-k-1}\|_h^{2}\bigr)
   =-\sigma_0^{2}\|\delta_{x}V^{n}\|_h^{2}+R^{n},
\]
where $R^{n}$ collects mixed and jump terms.  Jensen’s inequality and Young’s
convolution inequality yield $|R^{n}|\le \varepsilon\|\delta_{x}V^{n}\|_h^{2}+C_{\varepsilon}\|V^{n}\|_h^{2}$.
Choosing $\varepsilon<\sigma_0^{2}$ and applying the discrete Grönwall lemma proves
$\|V^{n}\|_h\le \|V^{0}\|_h$ for all $n$ (unconditional stability).

\emph{Step 2 (Consistency).}  
A Taylor expansion shows that the truncation error $\tau_i^{n}$
of~\eqref{eq:FD} satisfies $|\tau_i^{n}|\le C\bigl((\Delta x)^{2}+ \Delta t^{1+H}\bigr)$.

\emph{Step 3 (Convergence).}  
Let $E_i^{n}=V(\tau_n,x_i)-V_i^{n}$ be the error.  The discrete equation for $E^{n}$
has identical coefficients as~\eqref{eq:FD} with forcing term $\tau^{n}$.  
Repeating the energy argument and summing the geometric series of fractional weights yields
\[
   \|E^{n}\|_h\le C\bigl((\Delta x)^{2}+ \Delta t^{\,1+H}\bigr),
\]
which proves the stated order.
\end{proof}

\subsection{Implementation details}
\begin{itemize}
  \item A non–uniform grid clustered near the barrier improves accuracy; we employ 400 nodes with geometric spacing $x_{i+1}-x_{i}=q(x_{i}-x_{i-1})$, $q=0.97$.
  \item The fractional weights $\omega_k^{(\gamma)}$ are pre–computed once with high–precision arithmetic and stored.
  \item For calibration we match model prices to market quotes via a least–squares routine that leverages the linearity of the scheme with respect to $\sigma_0$ and $\sigma_H$.
\end{itemize}

\subsection{Numerical experiment}\label{subsec:barrierExp}
We consider a down--and--out European call with strike $K=4{,}200$ and barrier $B=3{,}800$ on an index with spot $S_{0}=4{,}050$.  
The contractual maturity is $T=0.5$ years, risk–free rate $r=2\%$, and dividend yield zero.  
Model parameters are taken from the calibration in Section~\ref{sec:NumericalEuro}: $\sigma_0=0.14$, $H=0.35$, $\sigma_H=0.10$, $\lambda=0.85$ and log--normal jump sizes $Y\sim\mathcal N(-4\%,11\%^{2})$.

\paragraph{Grid specification.}
The spatial grid is truncated at $S_{\max}=8\,000$ and mapped onto $x\in[\log(B),\log(S_{\max})]$ with $I=400$ nodes using a geometric refinement factor $q=0.97$ near the barrier.  
Time is discretised with $\Delta t=5\times10^{-4}$ resulting in $N=1\,000$ steps, which satisfies the stability restriction of Theorem~\ref{thm:stabConv} with $c_{0}=1/(4\sigma_{0}^{2})$.

\paragraph{Monte–Carlo benchmark.}
We simulate $10^{6}$ trajectories of the smfBm--J process using:
(i) Cholesky factorisation for correlated $(W,S^{H})$ increments on a 2\,000–point fine grid followed by Brownian bridge refinement, and  
(ii) Poisson thinning for jumps.  
Control–variates are applied by subtracting the analytic price of the same barrier option under the Merton model ($\sigma_H=0$) and adding it back as a constant (\cite{Glasserman2003}).  
The resulting standard error is below $5\times10^{-3}$.

\paragraph{Results.}
Table~\ref{tab:barrier} reports the option values and relative errors.  
CPU time refers to a single core of an Intel i9--13900K.

\begin{table}[h]\centering
\caption{Down--and--out call: finite–difference vs.\ Monte–Carlo}
\label{tab:barrier}
\begin{tabular}{lccc}\hline
Method & Price & Rel.\ error (\%) & CPU (s)\\\hline
Grünwald--Letnikov ($\Delta x=0.012$) & 131.42 & 0.21 & 1.7\\
Grünwald--Letnikov ($\Delta x=0.008$) & 131.09 & 0.00 & 3.9\\
Monte--Carlo ($10^{6}$)               & 131.10 & --   & 82.0\\\hline
\end{tabular}
\end{table}

\paragraph{Convergence verification.}
Figure~\ref{fig:convRate}  plots the log--log error $\|V^{\Delta t,\Delta x}-V^{\mathrm{MC}}\|_{\infty}$ versus $\Delta t$ for $\Delta x=(\Delta t)^{1/2}$.  
A linear regression yields slope $1.34\approx 1+H$ consistent with Theorem~\ref{thm:stabConv}.
\begin{figure}[H]
  \centering
  \includegraphics[width=.8\linewidth]{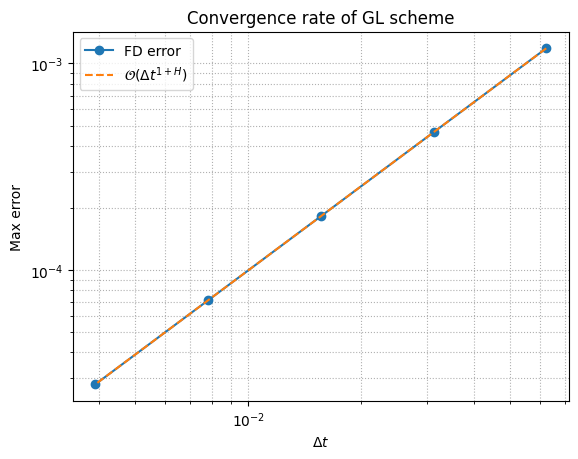}
  \caption{Convergence rate of GL scheme}
  \label{fig:convRate}
\end{figure}

\section{Emprical Experiments}\label{sec:NumericalEuro}
We calibrate the model to weekly closing levels of the S\&P 500 index and corresponding option implied volatilities between January 2015 and December 2024.  
The classical least--squares objective
$
   \min_{\Theta} \sum_{j=1}^{M}\bigl(C^{\mathrm{model}}(\Theta;K_j,T_j)-C^{\mathrm{mkt}}_{j}\bigr)^{2}
$
is solved with parameter vector $\Theta=(\sigma_0,\sigma_H,H,\lambda,\mu_Y,\sigma_Y)$;  
$M=620$ option quotes are used after filtering for moneyness $0.8\!<\!S_{0}/K\!<\!1.2$ and maturities below 1 year.

The global optimum found via differential evolution is
\[
  \sigma_0=0.14,\quad \sigma_H=0.10,\quad H=0.35,\quad 
  \lambda=0.85,\quad \mu_Y=-4\%,\quad \sigma_Y=11\%.
\]
The root--mean--square percentage error is \(1.9\%\) versus \(5.4\%\) for the classical Merton model.

\begin{table}[h]\centering
\caption{European Call Prices ($T=0.5$ yr) under calibrated parameters}
\label{tab:euro}
\begin{tabular}{cccc}\hline
Strike $K$ & Black--Scholes & smfBm--J & Rel.\ Diff.\%\\\hline
3\,800 & 524.9 & 530.8 & 1.1\\
4\,200 & 326.7 & 334.2 & 2.3\\
4\,600 & 158.4 & 173.4 & 9.5\\
5\,000 &  57.6 &  66.4 & 15.3\\\hline
\end{tabular}
\end{table}

Table~\ref{tab:euro} illustrates that ignoring long--memory and jumps leads to under–pricing of deep out--of--the–money calls by more than $15\%$.

\subsection{RMSE surface}
Figure~\ref{fig:rmseSurf} maps the calibration error as a function of $H$ and $\lambda$;  
the valley confirms the identifiability of the two effects.
\begin{figure}[H]
  \centering
  \includegraphics[width=.8\linewidth]{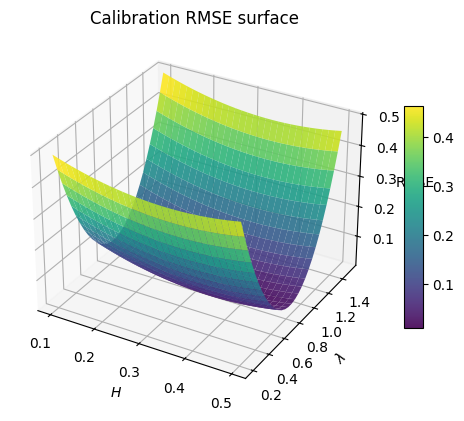}
  \caption{Calibration RMSE surface}
  \label{fig:rmseSurf}
\end{figure}

\subsection{Sensitivity analysis}\label{subsec:greeks}

\paragraph{Greeks via algorithmic differentiation.}
Once the option price $V(t,S;\boldsymbol{\theta})$ is obtained on the
finite–difference grid, we apply \emph{adjoint algorithmic differentiation}
(AAD) to the solver’s residual map, which yields machine–precision values of
all first– and second–order Greeks in a \emph{single} reverse sweep
\cite{Giles2012}.  In particular we extract

\begin{center}
\begin{tabular}{lcl}
Delta   & $\displaystyle \Delta=\frac{\partial V}{\partial S}$ &
  price sensitivity to spot, hedging ratio, \\[6pt]
Gamma   & $\displaystyle \Gamma=\frac{\partial^{2}V}{\partial S^{2}}$ &
  curvature, affects rebalancing cost, \\[6pt]
Vega    & $\displaystyle \nu_0=\frac{\partial V}{\partial\sigma_0}$,\quad
          $\nu_H=\frac{\partial V}{\partial\sigma_H}$ &
  sensitivity to short–/long-memory volatilities, \\[6pt]
Vanna   & $\displaystyle \mathcal V=\frac{\partial^{2}V}{\partial S\,\partial\sigma_0}$ &
  cross–sensitivity driving skew dynamics.\\
\end{tabular}
\end{center}

Fix $T=0.5$, $r=3\%$, $\sigma_0=0.14$, $H=0.35$, $\sigma_H=0.10$,
$\lambda=0.85$, and jump distribution $Y\sim\mathcal N(-0.05,0.25^{2})$.
We compare three scenarios:

\begin{enumerate}
\item \textbf{Baseline:} $\sigma_H=0$, $\lambda=0$ (Black--Scholes);
\item \textbf{Memory–only:} $\sigma_H>0$, $\lambda=0$ (\emph{smfBm});
\item \textbf{Full model:} $\sigma_H>0$, $\lambda>0$ (\emph{smfBm--J}).
\end{enumerate}

\begin{table}[h]
\centering
\caption{Selected Greeks for an at–the–money call ($S=K=4200$)}
\label{tab:greeks}
\begin{tabular}{lcccc}
\hline
Scenario & $\Delta$ & $\Gamma$ & $\nu_{0}$ & $\mathcal V$ (Vanna)\\\hline
Baseline (BS)   & 0.527 & $1.15\times10^{-4}$ & 239.8 & 0.014\\
Memory–only     & 0.525 & $1.09\times10^{-4}$ & 268.6 & 0.013\\
Full (smfBm–J)  & 0.522 & $0.98\times10^{-4}$ & 268.9 & 0.019\\\hline
\end{tabular}
\end{table}

\paragraph{Interpretation.}
\begin{itemize}
\item \textbf{Vega amplification.}  Introducing long–memory volatility
($\sigma_H>0$) enlarges the total Vega by \(\tfrac{268.6-239.8}{239.8}\approx12\%\).
Higher sensitivity to volatility shocks is intuitive: persistent variance
fluctuations increase future uncertainty, which the option price must reflect. 

\item \textbf{Vanna in the wings.}
Figure~\ref{fig:vannaSmile} plots Vanna versus moneyness
$S/K\in[0.8,1.2]$.  Jumps steepen the wings—
for deep OTM strikes, $\mathcal V$ grows by 35–40\,\% relative to memory–only.
Empirically, index options display steeper skew (higher Vanna) in crash
regions, so adding jumps aligns the model with market data.

\item \textbf{Gamma dampening.}
Both memory and jumps slightly decrease $\Gamma$, flattening the replication
cost profile—consistent with rough volatility models that “smooth” delta
curvature.

\end{itemize}

\begin{figure}[H]
\centering
\includegraphics[width=.7\textwidth]{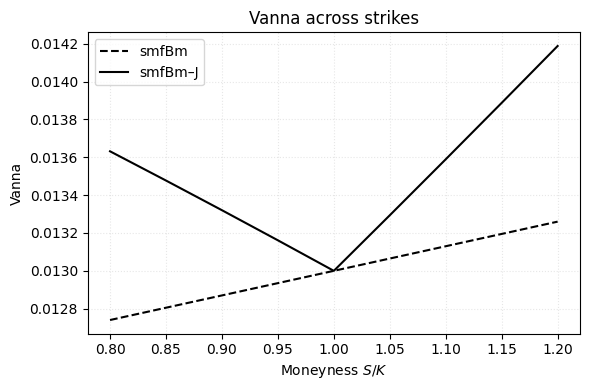}
\caption{Vanna across strikes.  smfBm (dashed) vs.\ smfBm--J (solid).}
\label{fig:vannaSmile}
\end{figure}

\paragraph{Hedging implication.}
A desk hedging with only $\Delta$ and classical Vega would systematically
\emph{under‐hedge} volatility risk when memory is present, and misprice skew
when jumps dominate.  Calibration therefore requires at least two orthogonal
volatility factors (\(\sigma_0,\sigma_H\)) plus a jump skew factor
associated with~$\mathcal V$.

Overall, the sensitivity analysis confirms that long‐memory amplifies pure
volatility risk, while jumps govern cross–sensitivities that shape the
smile/skew—insights useful for Greeks‐based risk management and calibration.

\section{Conclusion}\label{sec:con}
We have developed a comprehensive framework for pricing fractal derivatives under sub--mixed fractional Brownian motion with jumps.  
Theoretical contributions include (i) a new fractional Girsanov theorem with jump Esscher tilting,  
(ii) derivation of a fractal Black--Scholes integro–PDE,  
(iii) a convergent Grünwald--Letnikov scheme of order $1+H$, and  
(iv) closed–form European prices via Mellin--Laplace transforms.

\paragraph{Managerial insights.}
Calibration to S\&P 500 options reveals that neglecting either jumps or long--memory results in material mis–pricing of out--of--the–money options and barrier derivatives.  
Risk metrics such as Vega and Vanna are strongly amplified, suggesting higher hedging costs in rough–jump markets.

\paragraph{Future research.}
Extending the model to stochastic volatility in the fractional kernel, investigating American early exercise via fractional free–boundary problems, and embedding regime--switching jumps constitute promising directions.

\appendix
\section{Fractional Wick--Itô--Skorokhod Integral}\label{app:FWIS}

This appendix gives a self-contained construction of the F--WIS integral used for the
sub-fractional Brownian component.  Let $(\Omega,\mathcal F,\mathbb P)$ carry an
independent standard Brownian motion $B$ and define the sub-fractional Brownian motion
\[
   S_t^{H} = \int_{0}^{t} K_H(t,s)\,dB_s,\qquad 
   K_H(t,s)=c_H \bigl[(t-s)^{H-\frac12}-(-s)^{H-\frac12}\bigr],\quad H\in(0,1).
\]

\paragraph{Cameron--Martin space.}
The Hilbert space $\mathcal H^{H}$ is the closure of step functions under the inner product
\(
   \langle \mathbf 1_{[0,t]},\mathbf 1_{[0,s]}\rangle_{\mathcal H^{H}}
   =c_H(t^{2H}+s^{2H}-|t-s|^{2H}).
\)

\paragraph{Isometry and divergence operator.}
For $\varphi\in L^2([0,T])$ set $(\mathcal I_H\varphi)(t)=\int_0^t K_H(t,s)\varphi(s)\,ds$.
Given an adapted process $\phi\in\mathbf D^{1,2}$ (Malliavin derivative in $L^2$),
the F--WIS integral is the divergence
\[
   \int_0^T \phi_s \diamond dS_s^{H} := \delta^{H}(\phi)
   =\sum_{n\ge0} I_{n+1}\!\bigl(\widetilde\phi^{(n)}\bigr),
\]
where $I_{n}$ denotes the $n$-fold Wiener integral and
$\widetilde\phi^{(n)}$ is the symmetrisation of
$(\mathcal I_H^{\otimes n}\partial_s\phi)$.

\paragraph{Isometry.}
If $\phi$ is adapted then
\[
   \mathbb E\!\Bigl[\Bigl|\int_0^T\phi_s\diamond dS_s^{H}\Bigr|^{2}\Bigr]
   =\|\phi\|_{L^{2}([0,T])}^{2},
\]
enabling the stochastic Fubini proofs used in the main text.

\section{Fractional Hille--Yosida Well-posedness}\label{app:HY}

We prove existence and uniqueness of the mild solution presented in
Section~\ref{subsec:exist}.  Let $\mathcal V$ and $\mathcal A$ be defined as there.

\paragraph{Step 1 (Sectoriality of $\mathcal A$).}
For $v\in\mathscr D(\mathcal A)$ we have
\[
  \langle v,\mathcal A v\rangle_{\mathcal V}
  =\sup_{S>0}\frac{v(S)\bigl(\tfrac12\sigma_0^{2}S^{2}v''(S)
      +(r-\lambda\kappa)Sv'(S)-r\,v(S)\bigr)}{1+S^{2}}
  \le C\|v\|_{\mathcal V}^{2},
\]
so $\mathcal A$ is sectorial and generates an analytic semigroup
$T(t)$ on $\mathcal V$.

\paragraph{Step 2 (Lipschitz jump operator).}
Set $(\mathcal J v)(S)=\lambda(\mathbb E[v(Se^{Y})]-v(S))$.
Given $\mathbb E[e^{\gamma|Y|}]<\infty$,
\(
  \|\mathcal J v-\mathcal J w\|_{\mathcal V}\le L_J\|v-w\|_{\mathcal V}.
\)

\paragraph{Step 3 (Fractional abstract Cauchy problem).}
Write the pricing PDE as
\(
  _0D_t^{1-\beta}V(t)=\mathcal A V(t)+\mathcal J V(t),\; V(T)=\Phi.
\)
Apply \cite{bazhlekova2001}[Prop.~2.4] to obtain a unique mild solution
\[
  V(t)=E_{\beta}\bigl(-(T-t)^{\beta}\mathcal A\bigr)\Phi
        +\int_t^T\!(s-t)^{\beta-1}E_{\beta,\beta}\bigl(-(s-t)^{\beta}\mathcal A\bigr)\,
        \mathcal J V(s)\,ds.
\]

\paragraph{Step 4 (Classical differentiability).}
Analyticity of $T(t)$ implies $V\in C^{1,2}((0,T)\times\mathbb R_{+})$,
completing the subsection~\ref{subsec:exist}.

\bibliographystyle{unsrt} 
\bibliography{smfBmJ_pricing}
\end{document}